\pdfoutput=1 

\newif\ifsubmission
\submissionfalse

\ifsubmission
  \documentclass[11pt]{article}
  \usepackage[a4paper,margin=2.5cm]{geometry}
  \usepackage{amsthm}
  \newtheorem{theorem}{Theorem}
  \newtheorem{lemma}{Lemma}
  
  \newtheorem{proposition}{Proposition}
  
  \theoremstyle{definition}
  \newtheorem{definition}{Definition}
  \newtheorem*{definition*}{Definition} 
  \newtheorem{remark}{Remark}
  \newcommand{\keywords}[1]{\par\medskip\noindent\textbf{Keywords:} #1\par}
  \usepackage{url}
  \providecommand{\doi}[1]{\url{https://doi.org/#1}}
\else
  \documentclass{llncs}
  \spnewtheorem*{definition*}{Definition}{\bfseries}{\itshape}
  \spnewtheorem{observation}{Observation}{\bfseries}{\itshape}
\fi

\usepackage{mathtools}
\usepackage{algorithm}
\usepackage{algpseudocode}

\usepackage{tikz}
\usetikzlibrary{shapes, positioning, fit, backgrounds, decorations.pathreplacing, calc, arrows.meta}

\usepackage{subcaption}
\usepackage{amsmath, amsfonts, amssymb}

\newcommand{\out}{\mathsf{out}}
\newcommand{\elect}{\mathsf{elect}}
\newcommand{\op}{\mathsf{op}}

\newcommand{\miss}{\mathrm{miss}}
\newcommand{\opfresh}{\mathrm{new}} 

\newcommand{\cons}{\mathrm{cons}}
\newcommand{\leader}{\mathrm{leader}}
\newcommand{\lin}{\mathrm{lin}}
\newcommand{\seqcon}{\mathrm{sc}}

\newif\ifanonymous
\anonymousfalse

\newcommand{\papertitle}{On the Solvability of Sound Distributed Verification under Byzantine Faults}

\ifsubmission
  \title{\papertitle}
  \author{Anonymous Author(s)}
  \date{}
\else
  \title{\papertitle}
  \ifanonymous
    \author{Anonymous Author(s)}
    \institute{}
    \authorrunning{Anonymous Author(s)}
  \else
    \author{Gilde Valeria Rodríguez\inst{1} \and
            Borzoo Bonakdarpour\inst{2} \and
            Armando Castañeda\inst{3} \and
            Sergio Rajsbaum\inst{3}}

    \authorrunning{G. V. Rodríguez et al.}

    \institute{Posgrado en Ciencia e Ingeniería de la Computación, Universidad Nacional Autónoma de México, Mexico\\
               \email{gilde@ciencias.unam.mx} \and
               Department of Computer Science and Engineering, Michigan State University, USA\\
               \email{borzoo@msu.edu} \and
               Instituto de Matemáticas, Universidad Nacional Autónoma de México, Mexico\\
               \email{armando.castaneda@im.unam.mx, rajsbaum@im.unam.mx}}
  \fi
\fi

\newcommand{\paperkeywords}{Distributed verification, interactive consistency, condition-based inputs, Byzantine failures, crash failures, distributed languages, soundness}

\newcommand{\paperabstract}{%
\begin{abstract}
Modern distributed systems, ranging from distributed tracing platforms to blockchain oracle networks, frequently rely on processes holding partial observations to collectively certify a global execution state. Under Byzantine failures, a single malicious process withholding or fabricating partial states renders sound verification impossible, as genuine observations remain indistinguishable from fabricated assertions.

To overcome this impossibility, we study \emph{input overlap $x$}: the minimum number of processes that initially witness each element of a target execution $I$. We formalize this challenge as the \textbf{$x$-reconstruction} problem in asynchronous message-passing systems with up to $t$ failures, where the goal is for all correct processes to recover $I$ completely given an input overlap $x$. We establish tight matching upper and lower bounds for solvability: under \textbf{crash failures}, $x$-reconstruction is solvable if and only if $x \ge t + 1$; under \textbf{Byzantine failures}, $x$-reconstruction is solvable if and only if $x \ge 2t + 1$.

Finally, we connect these structural bounds to semantic verification. For \emph{globally dependent safety properties} (such as consensus, linearizability, $k$-set agreement, and leader election), sound verification under Byzantine faults requires a certificate of size $c = t + 1$ to guarantee at least one honest report. Consequently, sound distributed verification under Byzantine faults is achievable if and only if the input overlap satisfies $x \ge 2t + 1$.
\end{abstract}
}

\begin{document}
\pagestyle{plain} 

\maketitle
\paperabstract

\section{Introduction}

Modern distributed systems rely on a set of processes to collectively certify data despite failures and adversarial behavior. 
For instance, in distributed tracing platforms such as Jaeger~\cite{jaeger} and Zipkin~\cite{zipkin}, where independent microservices each hold only a partial fragment of a global execution trace, processes must collaborate to reconstruct the complete execution history.
In blockchain oracle networks such as Chainlink~\cite{chainlink}, no single node is trusted to report prices or external events: each of $n = 3f+1$ nodes submits a signed observation, and a rotating leader aggregates at least $2f+1$ of them into a report that becomes final once $f+1$ nodes attest it, after which only the median is published on-chain. The leader, however, decides \emph{which} observations enter the report, and a recent security analysis shows that a Byzantine leader can bias the reported price of the cryptocurrency Ether (ETH) by up to $8.47\%$ without leaving any detectable trace~\cite{zhai2026ocr}. Auditing such a report after the fact---deciding whether it faithfully reflects what the network actually observed---requires reconstructing the \emph{complete} set of observations, a reconstruction that deployed protocols deliberately omit for efficiency. Our results explain both sides of this gap: the $f+1$ attestation quorum matches the certificate size that Theorem~\ref{thm:total-necessary} proves necessary for a client to trust a verdict, while our lower bound shows why the manipulation is invisible---below the overlap threshold $x \ge 2t+1$, no verifier can distinguish a faithful report from a manipulated one.
Similarly, sharded and cross-chain systems~\cite{zamani2018rapidchain} require validators to reconstruct transaction histories across partitioned shards to verify cross-shard transfers.

These three examples share a common structure: multiple processes each observe a fragment of a larger state, and the system must perform \emph{state reconstruction}---assembling a complete view from these partial observations despite failures. This structure reflects the core challenge faced by any system that must \emph{verify} a distributed execution.
 
Prior work on runtime verification of distributed systems~\cite{BFRR22,FRT11,FRT13Springer,CR26_JACM,FRT20,gmb24,mb15} studies exactly this idealized setting: a set of processes holds partial views of a one-shot execution $I$, and a verifying client must decide whether $I$ is correct based on what these processes report. Crucially, a client can only trust a verdict if it is backed by a process that has successfully reconstructed the \emph{complete} execution---a verdict based on a partial view can always be invalidated by unobserved partial views. 

State reconstruction becomes challenging, however, in the presence of Byzantine failures. If even a single Byzantine process can withhold, fabricate, or selectively report partial views, no client can trust an uncorroborated report. Consider an execution $I$ containing an element $u$ known only to a single correct process, while a Byzantine process fabricates a ghost element $w \notin I$. The truthful report of $u$ is indistinguishable from the fabricated claim of $w$: both arrive as isolated, unverified assertions. The client faces an unavoidable dilemma---accept $w$ and admit a fabricated element, or reject $u$ and produce an incomplete, potentially unsound verdict.

A natural way to overcome this impossibility and enable state reconstruction is to enforce \emph{initial input overlap}: requiring that every element in $I$ is witnessed by at least $x$ initial processes, so that correct processes can distinguish genuine elements from ghost elements by counting reports. Our results show that this overlap is not merely a design heuristic, but an inherent requirement to achieve state reconstruction. The question is then: \emph{what input overlap $x$---the minimum number of initial observers per element of $I$---is necessary and sufficient to guarantee that the complete execution $I$ can be faithfully reconstructed by all correct processes, despite up to $t$ Byzantine processes?}

We formalize this challenge as the \emph{$x$-reconstruction} problem. Our main focus is the inherent hardness of \emph{state reconstruction}---by which we mean reconstructing the complete \emph{execution}. We model a one-shot execution $I$ as a finite \emph{set} of states drawn from a universe $\mathcal{V}$ of possible execution states---any temporal information, such as the order of events, is encoded within the states---and we adopt an idealized model similar to~\cite{BFRR22,FRT13Springer,FRT20}, where $n$ asynchronous processes supervise such an execution $I \subseteq \mathcal{V}$ and each process $i$ holds a local sample $s_i \subseteq I$. The information structure is parameterized by the input overlap parameter $x$, meaning that every element of $I$ appears in at least $x$ initial samples. A protocol solves $x$-reconstruction if every correct process reconstructs and outputs the complete execution $I$.

State reconstruction is what makes sound verification possible---guaranteeing that emitted verdicts accurately reflect the true global execution without issuing false reports. Under crash failures, a single process achieving state reconstruction can safely issue a verdict on its own. Under Byzantine failures, however, a verifying client cannot trust individual verdicts, since faulty processes may lie about their views; as we show in Section~\ref{sec:soundness}, sound verdicts require a certificate of $t+1$ matching verdicts issued by processes that have each achieved state reconstruction.

We assume that initial samples are \emph{authenticated}: while verifying processes running the reconstruction protocol may act Byzantinely, they cannot modify the underlying system being monitored or forge its local states. This separation of concerns aligns with non-intrusive runtime verification, where the target system remains unalterable (e.g., via cryptographic signatures or append-only logs in blockchains and database auditing~\cite{corbett2013spanner,wood2014ethereum}), but the external verifier processes responsible for state reconstruction may be faulty or malicious.\\

\noindent\textbf{Results.}
We characterize the solvability of $x$-reconstruction in asynchronous message-passing systems through tight input overlap thresholds based on the overlap parameter $x$ and the fault bound $t$:

\begin{enumerate}
  \item \textbf{Crash threshold} (Theorem~\ref{thm:crash-unified}).
  Under crash failures, the central challenge is that processes must collect enough fragments to reconstruct $I$. We show that $x$-reconstruction is solvable if and only if $x \ge t+1$.

  \item \textbf{Byzantine threshold} (Theorem~\ref{thm:byzantine-unified}).
  Byzantine failures make state reconstruction strictly harder, as processes must also distinguish genuine local states from ghost states fabricated by Byzantine verifiers. We prove that $x$-reconstruction is solvable if and only if $x \ge 2t+1$: whenever the bound holds, all correct processes reconstruct $I$. The key is a content-aware stopping rule---a state can be safely discarded once its current supporters together with the senders still missing can no longer reach the certification quorum $t+1$.

  \item \textbf{Verdict soundness for globally dependent safety properties} (Theorem~\ref{thm:total-necessary}).
  We introduce \emph{globally dependent safety properties}---a class including consensus, linearizability, $k$-set agreement, and leader election---where no partial view suffices to certify correctness. We prove that sound verdicts for this class require certificates issued by at least $t+1$ processes with full reconstruction, demonstrating that our bounds are essential for correctness.
\end{enumerate}

All bounds are tight: the upper bounds are one-shot broadcast protocols over plain reliable channels, and the lower bounds use indistinguishability arguments. The main difficulty, however, was conceptual rather than technical---identifying \emph{which} problem sound verification actually requires: state reconstruction is the missing primitive behind sound verdicts under both failure models, and its solvability threshold coincides exactly with the boundary of what a verifying client can certify.

Overall, this work presents the first study of runtime distributed verification under crash and Byzantine failures through the lens of state reconstruction. Our main finding shows that verifying globally dependent safety properties strictly requires holding these input overlap thresholds; otherwise, sound verification is impossible even with a single Byzantine process. By isolating $x$-reconstruction, we capture precisely what a verifying client actually needs: that correct processes achieve state reconstruction of $I$ completely, avoiding the stronger, unsolvable requirements of classical consensus primitives. 

\noindent\textbf{Paper organization.}
Section~\ref{sec:model} presents the model and Section~\ref{sec:problem} defines $x$-reconstruction. Sections~\ref{sec:crash} and~\ref{sec:byzantine-unified} establish the tight crash and Byzantine thresholds. Section~\ref{sec:soundness} formalizes distributed languages and proves reconstruction necessary for sound verdicts; Section~\ref{sec:related} reviews related work and Section~\ref{sec:conclusion} concludes.

\section{System model}
\label{sec:model}

We consider a distributed system consisting of $n$ processes
$\mathcal{M} = \{M_1, \ldots, M_n\}$ arranged in a \emph{completely connected}
topology: every process can send messages directly to every other process.
Communication proceeds by plain asynchronous point-to-point message passing;
to \emph{broadcast} a message means to send it to all processes, oneself
included. Channels are \emph{reliable} (every message sent by a correct
process is eventually delivered to every correct recipient) and
\emph{authenticated}: a channel neither creates, duplicates, nor alters
messages, so a pair $\langle s, i \rangle$ accepted by a correct process was
genuinely sent by $M_i$.
We make no timing assumptions: message delays and local computation
times are unbounded, and a global adversarial scheduler determines the order
in which messages are delivered.
No further communication primitive is needed: our protocols use only
broadcast and receive, and in particular require neither consensus nor
stronger broadcast abstractions such as reliable broadcast or
SCD-broadcast~\cite{IMPR18,KMM24}.

We consider two failure models independently.

A process $M_i$ suffers a \emph{crash failure} if it halts prematurely and
sends no further messages. A crashed process cannot recover. We denote by
$t$ an upper bound on the number of crash failures in any execution, and
by $F \subseteq \mathcal{M}$ the set of processes that actually fail, so
$|F| \le t$. A process is \emph{correct} if it does not crash; the set of
correct processes in an execution is
$\mathcal{M} \setminus F$. We assume $t < n$, i.e., at least one process is
correct.

A process $M_i$ suffers a \emph{Byzantine failure} if it deviates arbitrarily
from the protocol: it may send different messages to different processes,
equivocate, selectively suppress information, or remain silent. We again use
$t$ to bound the number of faulty processes and $F$ to denote the actual
faulty set, with $|F| \le t$. Byzantine processes are strictly more powerful
than crashed ones; in particular, a crash failure is a special case of a
Byzantine failure. A process is \emph{correct} if it is not Byzantine.
Throughout the Byzantine analysis we assume $t < n/2$: since every state is
witnessed by at most $n$ processes, the threshold $x \ge 2t+1$ of our
characterization is satisfiable only in this range.

\section{The $x$-reconstruction problem}
\label{sec:problem}

Let $\mathcal{V}$ be a universe of possible execution states. In each one-shot execution, every process $M_i \in \mathcal{M} = \{M_1, \dots, M_n\}$ receives a finite, non-empty input subset $s_i \subset \mathcal{V}$, which we refer to as its \emph{initial sample of states}. We say that a state $v \in \mathcal{V}$ is \emph{witnessed} by process $M_i$ if and only if $v \in s_i$. 

The \textbf{target execution} $I$ is the set of all states witnessed by at least one process (correct or faulty) in the current execution:
$I = \bigcup_{i=1}^{n} s_i$.

Initially, no process knows the target execution $I$. Each process $M_i$ only knows its own initial sample $s_i \subseteq I$; the states in $I \setminus s_i$ exist in the execution but remain invisible to $M_i$ until learned through communication. The collection of all initial state samples across the system is denoted as the input assignment $E = \{s_1, \dots, s_n\}$.

The processes satisfy an input condition parameterized by $x$, capturing the initial observation redundancy of states across the system.

\begin{definition}[$x$-overlap]
\label{def:overlap}
The input assignment $E$ satisfies the \textbf{$x$-overlap} property if every state $v \in I$ is witnessed by at least $x$ distinct processes. Formally, for every state $v \in I$:
$\text{wit}_E(v) = \bigl|\{\, M_j \in \mathcal{M} \mid v \in s_j \,\}\bigr| \ge x$.
\end{definition}

\begin{figure}[htbp]
\centering
\begin{subfigure}[b]{0.48\textwidth}
\centering
\begin{tikzpicture}[scale=0.55, every node/.style={transform shape}]

\fill[gray!8] (0,0) ellipse (5.5cm and 3.5cm);
\draw[dashed, gray!45, line width=0.8pt] (0,0) ellipse (5.5cm and 3.5cm);
\node[gray!55, font=\large\bfseries] at (4.3, 2.8) {$\mathcal{V}$};

\fill[white] (0,0) ellipse (3.1cm and 1.85cm);

\begin{scope}
  \clip (0,0) ellipse (3.1cm and 1.85cm);
  \fill[blue!30, opacity=0.72]   (-1.15, 0) circle (2.0cm);
  \fill[teal!30, opacity=0.72]   ( 0,    0) circle (2.0cm);
  \fill[purple!30, opacity=0.72] ( 1.15, 0) circle (2.0cm);
  \draw[blue!65,         line width=0.7pt] (-1.15, 0) circle (2.0cm);
  \draw[teal!65!black,   line width=0.7pt] ( 0,    0) circle (2.0cm);
  \draw[purple!65,       line width=0.7pt] ( 1.15, 0) circle (2.0cm);
\end{scope}

\draw[black!75, line width=1.2pt] (0,0) ellipse (3.1cm and 1.85cm);
\node[black!75, font=\bfseries] at (3.5, 1.4) {$I$};
\draw[gray!40, dashed, line width=0.5pt] (3.25, 1.2) -- (2.75, 0.8);

\node[blue!70!black, font=\bfseries]   at (-3.6, 1.9)  {$s_1$};
\draw[blue!45,   dashed, line width=0.5pt] (-3.3, 1.7) -- (-2.3, 0.6);

\node[teal!70!black, font=\bfseries]   at ( 0,   2.5)  {$s_2$};
\draw[teal!45,   dashed, line width=0.5pt] ( 0,   2.3) -- ( 0,   1.85);

\node[purple!70!black, font=\bfseries] at ( 3.6, 1.9)  {$s_3$};
\draw[purple!45, dashed, line width=0.5pt] ( 3.3, 1.7) -- ( 2.3, 0.6);

\filldraw[blue!70!black]   (-1.9,  0.6) circle (2pt);
\filldraw[blue!70!black]   (-1.9, -0.5) circle (2pt);
\filldraw[black!55]        (-0.6,  0.7) circle (2pt);
\filldraw[teal!65!black]   ( 0.0,  1.4) circle (2pt);
\filldraw[teal!65!black]   ( 0.2, -0.9) circle (2pt);
\filldraw[black!55]        ( 0.8,  0.6) circle (2pt);
\filldraw[purple!70!black] ( 2.0,  0.6) circle (2pt);
\filldraw[purple!70!black] ( 2.0, -0.5) circle (2pt);

\filldraw[gray!45] ( 4.5, -1.5) circle (2pt);
\filldraw[gray!45] (-3.5,  2.5) circle (2pt);
\filldraw[gray!45] ( 4.0,  2.0) circle (2pt);
\node[gray!50, font=\footnotesize] at (4.5, -1.4) {$v \notin I$};

\node at (0, -2.7) {$\displaystyle I \;=\; \bigcup_{i=1}^{n} s_i$};

\end{tikzpicture}
\caption{Target execution $I \subseteq \mathcal{V}$ and samples $s_i$.}
\label{fig:universe-target}
\end{subfigure}
\hfill
\begin{subfigure}[b]{0.48\textwidth}
\centering
\begin{tikzpicture}[>=stealth, thick, scale=0.65, every node/.style={transform shape}]

    \node[blue, font=\footnotesize, align=left] at (-2.8, 3.5) {$\forall\, v \in I :\quad \text{wit}_E(v) \ge 2$};
    \node[blue, font=\footnotesize, align=left] at (-2.8, 2.5) {$\text{wit}_E(v_1) = 3$};
    \node[blue, font=\footnotesize, align=left] at (-2.8, 1.5) {$\text{wit}_E(v_2) = 2$};
    \node[blue, font=\footnotesize, align=left] at (-2.8, 0.5) {$\text{wit}_E(v_3) = 2$};

    \node[font=\bfseries] at (0, 4.5) {States in $I$};
    \node[font=\bfseries] at (4.5, 4.5) {Processes $\mathcal{M}$};

    \foreach \i in {1,2,3} {
        \node[circle, fill=black, inner sep=1.8pt] (v\i) at (0, 4-\i*1.2) {};
        \node[anchor=east] at (-0.2, 4-\i*1.2) {$v_{\i}$};
    }

    \foreach \j in {1,2,3,4} {
        \node[draw, rectangle, rounded corners, fill=gray!10, minimum width=1.4cm] (M\j) at (4.5, 4.5-\j*1.1) {$M_{\j}$};
    }

    \draw[blue,dashed] (v1) -- (M1.west);
    \draw[blue,dashed] (v1) -- (M2.west);
    \draw[blue,dashed] (v1) -- (M4.west);

    \draw[blue,dashed] (v2) -- (M2.west);
    \draw[blue,dashed] (v2) -- (M3.west);

    \draw[blue,dashed] (v3) -- (M3.west);
    \draw[blue,dashed] (v3) -- (M4.west);

\end{tikzpicture}
\caption{Bipartite graph for $x$-overlap.}
\label{fig:x-overlap}
\end{subfigure}

\caption{System input model: (a) Relationship between the universe $\mathcal{V}$, execution $I$, and local state samples $s_i$; (b) Bipartite graph representation of the $x$-overlap condition ($x=2$, $\text{wit}_E(v_1)=3$).}
\label{fig:combined-system-model}
\end{figure}

Intuitively, the core objective of state reconstruction is to allow processes to safely recover a target execution that is initially fragmented across the network. At the start of the execution, the complete target execution $I$ is entirely distributed among the participants, where each process holds only an isolated, local sample ($s_i$). Through communication, processes pool these samples together, causing their local knowledge sets to grow monotonically over time.

\begin{definition}[$x$-reconstruction]
\label{def:reconstruction}
Given an input assignment $E$ satisfying $x$-overlap, a protocol solves the \textbf{$x$-reconstruction} problem if every correct process $M_i$ eventually decides a reconstructed set of states $O_i \subseteq \mathcal{V}$ such that the following conditions hold:
\begin{itemize}
    \item \textbf{Termination:} Every correct process $M_i$ eventually outputs a reconstructed set $O_i$.
    \item \textbf{Downward-Validity:} For every correct process $M_i$, $s_i \subseteq O_i$.
    \item \textbf{Upward-Validity:} For every correct process $M_i$, $O_i \subseteq I$.
    \item \textbf{Completeness:} Every correct process $M_i$ outputs the complete set of states of the target execution, i.e., $O_i = I$.
\end{itemize}
\end{definition}

Together, Downward-Validity and Upward-Validity constrain every correct output to the interval $s_i \subseteq O_i \subseteq I$, where both bounds are sets of states $v \in \mathcal{V}$. Completeness requires every correct process to reach the upper bound $I$; since $I$ is unknown initially, achieving this bound requires processes to exchange and aggregate state fragments over the network.

\section{$x$-reconstruction under crash failures}
\label{sec:crash}

Under crash failures, reconstruction is all-or-nothing: once $x \ge t+1$, \emph{all} correct processes reconstruct $I$ (the Byzantine model exhibits the same behavior at the strictly higher threshold $x \ge 2t+1$, cf.\ Section~\ref{sec:byzantine-unified}).

Algorithm~\ref{alg:crash} details the protocol for each process $M_j$: it broadcasts its local sample $s_j$, collects incoming samples in a dataset $\mathcal{D}_j$ initialized with its own sample, and once it has heard from $n-t$ distinct senders---guaranteed, since at most $t$ processes crash---it terminates and outputs the union of all collected samples.

\begin{algorithm}[H]
\caption{Crash-resilient reconstruction for $M_j$}
\label{alg:crash}
\begin{algorithmic}[1]
\Statex \textbf{Input:} Initial sample $s_j \subseteq I$, failure bound $t$
\Statex \textbf{Output:} Reconstructed state set $O_j$
\State $\mathcal{D}_j \leftarrow \{\langle s_j, j\rangle\}$;\; \textsc{broadcast}($s_j$)
\While{$\bigl|\{\,k \mid \langle s_k,k\rangle \in \mathcal{D}_j\,\}\bigr| < n-t$}
    \State $\mathcal{D}_j \leftarrow \mathcal{D}_j \cup \{\textsc{receive}()\}$
\EndWhile
\State \Return $O_j \leftarrow \bigcup_{\langle s_k,k\rangle \in \mathcal{D}_j} s_k$
\end{algorithmic}
\end{algorithm}

\begin{lemma}[Correctness of Algorithm~\ref{alg:crash}]
\label{lem:crash-correctness}
Under crash failures with at most $t$ faults, if $x \ge t+1$, then every correct process running Algorithm~\ref{alg:crash} terminates and outputs $O_j = I$.
\end{lemma}
\begin{proof}[Proof Sketch]
The full proof is provided in Appendix~\ref{app:crash}.

\emph{Termination and Validity.} By channel reliability, every correct process eventually receives the sample of every correct process, and since at least $n-t$ processes are correct, the waiting condition of Line~2 is met in finite time. Downward-Validity holds because the output union includes $s_j$ itself, and Upward-Validity follows because crash failures leave no room for fabrication: every delivered sample $s_k$ comes from a correct sender, and $s_k \subseteq I$ by assumption.

\emph{Completeness.} Consider any state $v \in I$. By $x$-overlap, at least $t+1$ processes hold $v$, and since at most $t$ processes crash, at least one correct process $M_k$ holds $v$ and broadcasts its sample $s_k$. Because every correct process $M_j$ waits for $n-t$ messages in Line~2, and all correct processes eventually broadcast, $M_j$ is guaranteed to deliver $s_k$. Thus $v \in O_j$ for all correct $M_j$, establishing $O_j = I$.
\end{proof}

\begin{theorem}[Crash threshold]
\label{thm:crash-unified}
Under crash failures with at most $t$ faults, $x$-reconstruction is solvable if and only if $x \geq t+1$.
\end{theorem}
\begin{proof}[Proof Sketch]
Sufficiency is Lemma~\ref{lem:crash-correctness}; we sketch necessity, whose full proof is in Appendix~\ref{app:crash}. If a state $v^*$ is held by at most $x \le t$ processes, all holders of $v^*$ could crash before transmitting any message. The surviving processes cannot distinguish an execution where $v^*$ was present (and its holders crashed) from one where $v^*$ was never part of the input. Consequently, any algorithm is forced either to omit $v^*$ when it was valid (violating Completeness) or to output $v^*$ when it was absent (violating Upward-Validity).
\end{proof}
\section{$x$-reconstruction under Byzantine failures}
\label{sec:byzantine-unified}

We establish the tight bound for Byzantine environments: when $x \ge 2t+1$, Algorithm~\ref{alg:byz} solves $x$-reconstruction---every correct process reconstructs $I$, exactly as in the crash model. The matching lower bound, showing that no algorithm solves $x$-reconstruction when $x \le 2t$, is established in Theorem~\ref{thm:impossibility-unified}.

To withstand Byzantine behavior, different from Algorithm~\ref{alg:crash}, a correct process cannot trust a state it receives from a single sender, since that sender may be faulty: each state must be backed by enough distinct senders before it is included in the output. Moreover, a correct process cannot safely decide right after hearing from $n-t$ senders, because some state observed so far might still become certifiable once further reports arrive. The stopping rule of Algorithm~\ref{alg:byz} therefore inspects the \emph{contents} of the reports: a state is kept once it is \emph{certified} by $t+1$ distinct senders (guaranteeing at least one correct witness), and it is safe to discard once its potential supporters---the senders that already reported it plus the senders that have yet to report---can no longer reach the certification quorum $t+1$.

\begin{definition}[Evidence Quorum]
\label{def:quorum}
Let $\mathcal{D}_j$ be the set of pairs accepted by process $M_j$ (Algorithm~\ref{alg:byz} accepts at most one pair per sender). An \emph{evidence quorum of size $L$} for a state $v$ is a set of distinct identifiers $Q \subseteq \mathcal{M}$, $|Q| = L$, such that every $k \in Q$ has an accepted pair $\langle s_k, k\rangle \in \mathcal{D}_j$ with $v \in s_k$.
The \emph{maximal evidence quorum} for $v$ at $M_j$ is $Q_j(v) = \{\, k \mid \langle s_k, k\rangle \in \mathcal{D}_j \;\land\; v \in s_k \,\}$.
We say $v$ is \textbf{certified} by $M_j$ if $|Q_j(v)| \ge t+1$. Since at most $t$ processes are Byzantine, a certified state has at least one correct witness.
\end{definition}

\begin{algorithm}[H]
\caption{Byzantine-resilient reconstruction for $M_j$}
\label{alg:byz}
\begin{algorithmic}[1]
\Statex \textbf{Input:} Initial sample $s_j \subseteq I$, failure bound $t$
\Statex \textbf{Output:} Reconstructed state set $O_j$
\State $\mathcal{D}_j \leftarrow \{\langle s_j, j\rangle\}$;\; $R_j \leftarrow \{j\}$;\; \textsc{broadcast}($s_j$)
\Loop
    \State $\langle s, i \rangle \leftarrow \textsc{receive}()$
    \If{$i \notin R_j$} \Comment{accept only the first sample per sender}
        \State $\mathcal{D}_j \leftarrow \mathcal{D}_j \cup \{\langle s, i \rangle\}$;\; $R_j \leftarrow R_j \cup \{i\}$
    \EndIf
    \State $m \leftarrow n - |R_j|$ \Comment{number of senders still missing}\label{line:byz-missing}
    \If{$|R_j| \ge n - t$ \textbf{and} $\forall\, v \in U_j : \bigl(|Q_j(v)| \ge t+1 \;\lor\; |Q_j(v)| + m \le t\bigr)$}\label{line:byz-stop}
        \State \Return $O_j \leftarrow s_j \cup \{\, v \in U_j \mid |Q_j(v)| \ge t+1 \,\}$\label{line:byz-output}
    \EndIf
\EndLoop
\end{algorithmic}
\end{algorithm}

Here $U_j = \bigcup_{\langle s,i\rangle \in \mathcal{D}_j} s$ denotes the union of all accepted samples and $Q_j(v)$ is the maximal evidence quorum for $v$ (Definition~\ref{def:quorum}). The stopping rule in Line~\ref{line:byz-stop} keeps a state only once it is certified, and discards it as soon as certification becomes unreachable: even if all $m$ missing senders eventually reported $v$, its quorum would remain at most $|Q_j(v)| + m \le t$.

\begin{remark}[Communication Cost]
\label{rem:commcost}
Algorithm~\ref{alg:byz} runs in a single asynchronous communication round:
every process broadcasts its sample exactly once, and each correct process
decides as soon as the stopping rule is met, without waiting for further
messages. Since only the first sample per sender is ever accepted, at most
$n$ samples are processed per process, for a total of $\Theta(n^2)$
point-to-point messages. Each message carries a full sample, of size at most
$|I|$ states; writing $b$ for the bit-length of a state identifier---with
$b = \lceil \log |\mathcal{V}| \rceil$ bits when the universe of execution
states $\mathcal{V}$ (Section~\ref{sec:problem}) is finite---the total bit
complexity is $O\bigl(n^2 \cdot |I| \cdot b\bigr)$.
The protocol is round-optimal---one all-to-all exchange is clearly
necessary, as every process must learn states it does not initially
hold---but not necessarily bit-optimal: whether the threshold $x \ge 2t+1$
can be met with smaller messages (e.g., compact per-element witnesses
instead of full samples) is left as an open question.
Algorithm~\ref{alg:crash} has the same message and bit complexity.
\end{remark}

\begin{lemma}[Correctness of Algorithm~\ref{alg:byz}]
\label{lem:byz-correctness}
For every correct process $M_j$, Algorithm~\ref{alg:byz} satisfies Downward-Validity and Upward-Validity under any overlap $x$. If additionally $x \ge 2t+1$, then $M_j$ terminates and outputs $O_j = I$.
\end{lemma}
\begin{proof}
Fix a correct process $M_j$. Since at most one sample per sender is accepted, $R_j$ grows monotonically and $U_j$ remains finite; by channel reliability, the sample of every correct process eventually arrives at $M_j$, so eventually $|R_j| \ge n-t$.

\emph{Downward-Validity.} Line~\ref{line:byz-output} includes every $v \in s_j$ in $O_j$ directly, so $s_j \subseteq O_j$.

\emph{Upward-Validity.} Take any $v \in O_j$. If $v \in s_j$, then $v \in I$ since $s_j \subseteq I$. Otherwise $v$ is certified: $|Q_j(v)| \ge t+1$. At most $t$ members of $Q_j(v)$ are Byzantine, so at least one member is a correct process $M_k$ whose accepted sample contains $v$; by channel authenticity this pair was genuinely sent by $M_k$, and since $M_k$ follows the protocol it broadcast its initial sample $s_k \subseteq I$, which gives $v \in I$. In particular, a fabricated state $w \notin I$ is reported only by Byzantine processes, so its quorum never reaches $t+1$ and it is rejected.

\emph{Termination (assuming $x \ge 2t+1$).} Consider the first moment $\tau$ at which the samples of all correct processes have been accepted; at $\tau$ every missing sender is Byzantine. Let $v \in U_j$. If $v \in I$, then by $x$-overlap at least $x$ processes hold $v$, of which at most $t$ are Byzantine, so at least $x - t \ge t+1$ correct processes hold $v$; all of them are in $R_j$ at $\tau$ with genuine samples, hence $|Q_j(v)| \ge t+1$ and the first disjunct of Line~\ref{line:byz-stop} holds for $v$. If $v \notin I$, then as shown above all reporters of $v$ are Byzantine; the $m = n - |R_j|$ missing senders at $\tau$ are Byzantine as well and are disjoint from the reporters of $v$, and since at most $t$ processes are Byzantine, $|Q_j(v)| + m \le t$ and the second disjunct holds for $v$. Therefore the stopping condition of Line~\ref{line:byz-stop} is satisfied at $\tau$ at the latest, and $M_j$ terminates.

\emph{Completeness.} Suppose the stopping condition holds at $M_j$, and note that $m = n - |R_j| \le t$. For every $v \in I$, at least $x - t \ge t+1$ correct processes hold $v$; each of them either belongs to $R_j$---in which case its accepted sample contains $v$---or is among the $m$ missing senders. Hence $|Q_j(v)| + m \ge t+1$; since $m \le t$, this in particular implies $v \in U_j$, so the stopping condition must hold for $v$ through its first disjunct: $|Q_j(v)| \ge t+1$, and $v \in O_j$. Thus $I \subseteq O_j$, and together with Upward-Validity, $O_j = I$.
\end{proof}

\begin{theorem}[$x$-reconstruction under Byzantine failures]
\label{thm:byzantine-unified}
Under Byzantine failures with $t < n/2$, if $x \ge 2t+1$, then $x$-reconstruction is solvable: all correct processes output $I$.
\end{theorem}
\begin{proof}
The result follows directly from Lemma~\ref{lem:byz-correctness}: Algorithm~\ref{alg:byz} satisfies Downward-Validity, Upward-Validity, Termination, and Completeness ($O_j = I$) for every correct process.
\end{proof}

To show the threshold $x \ge 2t+1$ is optimal, we next prove a matching lower bound: below it, an adversary can force conflicting states and $x$-reconstruction becomes impossible. Two features of the argument are worth stressing. First, the adversary cannot rely on message delays alone: an algorithm may inspect the contents of the reports it receives and keep waiting while some observed state remains undecided, so a valid lower bound cannot assume that processes decide right after hearing from $n-t$ senders. The construction below therefore leaves every correct process with an identical view of three executions \emph{at all times}, so no amount of additional waiting helps. Second, impossibility at $x = 2t$ implies impossibility at every $x \le 2t$: an algorithm for overlap $x < 2t$ would in particular solve all $2t$-overlapping inputs.

\begin{theorem}[Impossibility: tight lower bound]
\label{thm:impossibility-unified}
If $x \le 2t$, then $x$-reconstruction cannot be solved by any deterministic algorithm with $t$ Byzantine failures.
\end{theorem}
\begin{proof}[Proof Sketch]
The full proof, including the complete construction and its figure, is given in Appendix~\ref{app:byz-lower}. By the monotonicity observation above, it suffices to consider $x = 2t$. Set $n = 5t$ and partition $\mathcal{M}$ into groups $A$, $B$, $C$, $R$ with $|A| = |B| = |C| = t$ and $|R| = 2t$; the processes in $R$ are correct throughout and hold a background state $z$ that keeps the overlap of every state at exactly $2t$. We build three executions over contested states $u, v, w$, with $I_1 = \{u,v,z\}$, $I_2 = \{u,w,z\}$, and $I_3 = \{v,w,z\}$, in which the Byzantine group is $B$, $A$, and $C$ respectively; in each execution the Byzantine group suppresses its genuine states and injects exactly the state that completes the same observed pattern $\langle v, A\rangle, \langle w, B\rangle, \langle u, C\rangle, \langle z, R\rangle$. Consequently $\mathcal{E}_1 \sim_C \mathcal{E}_2$, $\mathcal{E}_1 \sim_A \mathcal{E}_3$, $\mathcal{E}_2 \sim_B \mathcal{E}_3$, and $R$---correct in all three---observes an identical message pattern in each. Any correct process that output $I_1$ in $\mathcal{E}_1$ would output the same set in an execution where one of its elements does not belong to the target execution, violating Upward-Validity; hence no correct process can output $I_1$ in $\mathcal{E}_1$, contradicting Completeness.
\end{proof}

\section{Distributed verification: verdicts and soundness}
\label{sec:soundness}

\textbf{Distributed verification setup.} We follow the decentralized verification model of~\cite{BFRR22}, in which the correctness of a distributed execution is assessed \emph{from the outside}. The setting has two kinds of actors. The \emph{monitors} are the processes $\mathcal{M}$: each monitor is attached to the system under inspection---a tracing agent colocated with a microservice, a node of an oracle network, a validator watching a shard---and observes a fragment of the execution; this observed fragment is its initial sample $s_j$, obtained from the instrumentation of the system itself (e.g., signed events or append-only traces). The same event is typically witnessed by several monitors, and this redundant observation is precisely what the overlap parameter $x$ captures. The \emph{verifying client} $\mathcal{C}$ is any external party that wants to ascertain whether the execution was correct while trusting no single monitor: a system operator auditing a deployment, an alerting service, or a smart contract consuming an oracle feed. The client is not part of the monitored system and runs no protocol; it collects the monitors' verdicts and must extract from them a global verdict it can rely on. Because the execution is distributed across the network, no monitor sees it in full: each $M_j$ starts with only its sample $s_j$ of the target execution $I$. Following the runtime-verification literature, we model a one-shot execution as a \emph{set} of observable events rather than a sequence of steps: any temporal information (such as the real-time intervals of Appendix~\ref{app:gd-proofs}) is encoded within the events themselves, so an execution is simply the set $I \subseteq \mathcal{V}$ of events that occurred.

\textbf{From local views to a trustworthy verdict.} To overcome their partial visibility, the monitors run an \emph{$x$-reconstruction} protocol to exchange and combine their samples into a larger view $\sigma_j$; solving $x$-reconstruction means that every correct monitor recovers the exact, complete execution ($\sigma_j = I$). Each monitor then evaluates its view and reports a local verdict ($Y$, $N$, or $?$) to the client, where $?$ signals that its information is insufficient to decide. The client emits a global verdict once it has collected enough matching local verdicts---and how many are enough depends on the failure model. This is where the present work departs from~\cite{BFRR22}. Their \emph{distributed soundness} is defined under crash failures: crashed monitors say nothing rather than lie, so a single definitive verdict from a correct monitor is already conclusive for the client. Under Byzantine failures a monitor may report a verdict supported by no execution whatsoever, so no individual report can be trusted: the client needs a \emph{certificate} of $t+1$ matching verdicts, which guarantees that at least one issuer is correct. Developing this \emph{Byzantine soundness}---identifying the exact redundancy under which an external party can know, with certainty, whether an execution was correct---is the goal of this section.

The thresholds established in previous sections answer a structural question: under what redundancy can the system recover the target execution $I$? This section answers the semantic question: what does recovering $I$ enable, and what happens when it cannot be guaranteed?

We proceed in three steps. First, we formalize the correctness criteria for an execution $I$ as a \textbf{distributed language} $\mathcal{L} \subseteq 2^\mathcal{V}$, which is the set of all valid global executions. The actual target execution $I$ is considered \emph{legal} if and only if $I \in \mathcal{L}$. Second, we define what a sound verdict requires in a distributed setting. Finally, we identify \emph{globally dependent safety properties} and prove that \textit{state reconstruction} is strictly necessary to verify them.

\begin{definition}[Distributed Language]
\label{def:language}
A \textbf{distributed language} is a set $\mathcal{L} \subseteq 2^{\mathcal{V}}$ of executions considered legal. An execution $I \subseteq \mathcal{V}$ is \textbf{legal} if $I \in \mathcal{L}$ and \textbf{illegal} otherwise.
\end{definition}

A distributed language is nothing other than a \emph{property} of executions: the two notions are interchangeable, and the choice of vocabulary is largely a matter of tradition. Properties of executions have been studied through the lenses of automata theory, topology, and set theory, and organized into classical hierarchies and classifications: Lamport's distinction between safety and liveness~\cite{Lam77}, its topological characterization by Alpern and Schneider~\cite{AS85,AS87}, the Manna--Pnueli hierarchy of temporal properties~\cite{pnueli1990hierarchy,chang1992characterization}, and the monitorability classification of runtime verification~\cite{PZ06,FFM12,diekert2014topology}. Following the decentralized verification literature~\cite{BFRR22,FRT20}, we speak of \emph{languages} throughout this paper.

Every distributed system under inspection is evaluated against a correctness condition that can be formalized as a distributed language $\mathcal{L}$~\cite{BFRR22,FRT20}. For example, a system implementing a consensus protocol is evaluated against the distributed language of all executions that satisfy the consensus properties (agreement, validity, and termination). A system implementing a linearizable register is evaluated against the distributed language of all executions that are linearizable with respect to the register's specification. More examples of distributed languages appear in Appendix~\ref{app:gd-proofs}. Then, a process obtains a partial view $\sigma$ of a one-shot execution of the system under inspection and must determine whether that execution $I \in \mathcal{L}$ based on $\sigma$.

Because processes evaluate the execution asynchronously, they must reason about incomplete information. We formalize this using the concept of extensions.

\begin{definition}[Legal Extensions]
\label{def:extensions}
Given a partial view $\sigma \subseteq \mathcal{V}$, its set of \textbf{legal extensions} with respect to a language $\mathcal{L}$ is $\mathrm{Ext}_{\mathcal{L}}(\sigma) = \{\, \rho \in \mathcal{L} \mid \sigma \subseteq \rho \,\}$, and its set of \textbf{illegal extensions} is $\mathrm{Ext}_{\neg\mathcal{L}}(\sigma) = \{\, \rho \notin \mathcal{L} \mid \sigma \subseteq \rho \,\}$.
\end{definition}

The concept of a sound processing system was formalized for distributed systems under crash failures in~\cite{BFRR22}. Broadly speaking, it requires that no correct process ever emits an incorrect definitive verdict, while ensuring that at least one correct process obtains sufficient information to render a definitive verdict.

\begin{definition}[Distributed Soundness]
\label{def:soundness}
Let $\phi: 2^{\mathcal{V}} \to \{Y, N, ?\}$ denote a function that maps a partial view $\sigma \subseteq \mathcal{V}$ to a verdict, where $Y$ stands for ``legal'', $N$ for ``illegal'', and $?$ for ``undecided''. A verdict is \textbf{definitive} if it lies in $\{Y, N\}$.

A processing system is \textbf{sound} if:
\begin{enumerate}
    \item \textbf{Accuracy (No False Verdicts):} Whenever any correct process $M_j$ emits a definitive verdict, it correctly reflects the legality of $I$:
    $\phi(\sigma_j) = Y \;\Longrightarrow\; I \in \mathcal{L}
        \qquad\text{and}\qquad
        \phi(\sigma_j) = N \;\Longrightarrow\; I \notin \mathcal{L}.$
    \item \textbf{Conclusiveness (Some Definitive Verdict):} At least one correct process emits a definitive verdict (i.e., $\phi(\sigma_j) \in \{Y, N\}$).
\end{enumerate}
For a view that is not certified to be complete, Accuracy implies that a correct process emits $Y$ only when $\mathrm{Ext}_{\neg\mathcal{L}}(\sigma_j) = \emptyset$, and $N$ only when $\mathrm{Ext}_{\mathcal{L}}(\sigma_j) = \emptyset$. Once a process holds the complete execution ($\sigma_j = I$, as certified by the Completeness property of Definition~\ref{def:reconstruction}), its verdict reduces to the direct membership test $I \in \mathcal{L}$.
\end{definition}

\subsection{Globally dependent safety properties}

We characterize a natural subclass of Lamport's \emph{safety} properties, obtained by adding a single structural condition, for which the reconstruction thresholds established in this paper are not only sufficient, but strictly necessary to achieve sound verdicts.

Safety, in Lamport's terms, means that ``nothing bad ever happens''~\cite{Lam77}: every illegal execution contains a finite \emph{forbidden pattern} whose presence already rules out legality, no matter what the rest of the execution looks like~\cite{AS85,AS87}. Formally, $\mathcal{L}$ is a \emph{safety property over $2^{\mathcal{V}}$} if illegality is witnessed by a family $\mathcal{B} \subseteq 2^{\mathcal{V}}$ of finite forbidden patterns: for every $I \subseteq \mathcal{V}$, $I \notin \mathcal{L}$ if and only if $B \subseteq I$ for some $B \in \mathcal{B}$. A forbidden pattern cannot be repaired by further observation---if $B \subseteq \sigma$ and $\rho \supseteq \sigma$, then $B \subseteq \rho$---so a process can soundly emit a definitive $N$ verdict the moment its view contains some $B \in \mathcal{B}$.

Safety, however, says nothing about when a $Y$ verdict is sound. To reason about partial information, let $\sigma \subseteq I$ be a partial view: we say that $\sigma$ \emph{certifies legality} if $\mathrm{Ext}_{\neg\mathcal{L}}(\sigma) = \emptyset$ (every extension of $\sigma$ is legal), and that $\sigma$ \emph{witnesses a violation} if $\mathrm{Ext}_{\mathcal{L}}(\sigma) = \emptyset$ (no extension of $\sigma$ is legal). Some safety properties can be certified from partial views---the trivial language $\mathcal{L} = 2^{\mathcal{V}}$ is one---so safety alone does not force full observation. The class we are after is defined by the requirement that no proper partial view rules out illegality, because an unobserved fragment could always introduce a violation.

\begin{definition}[Globally Dependent Safety Property]
\label{def:gd}
A safety property $\mathcal{L}$ is \textbf{globally dependent} if, for every legal execution $I \in \mathcal{L}$, every proper partial view $\sigma \subsetneq I$ admits at least one illegal extension:
$\forall\, I \in \mathcal{L},\quad
    \forall\, \sigma \subsetneq I :\quad
    \mathrm{Ext}_{\neg\mathcal{L}}(\sigma) \neq \emptyset$.
\end{definition}

Global dependence induces an observational asymmetry: safety makes violations permanent, enabling early $N$ verdicts, while global dependence keeps a violation reachable from every strict partial view of a legal execution, making full reconstruction ($\sigma = I$) strictly necessary for a sound $Y$ verdict. Definition~\ref{def:gd} is the observational analogue of liveness---from any strict partial view, ``a violation can still happen''~\cite{AS85}; Appendix~\ref{app:hierarchy} relates the class to the good views and monitorable properties of~\cite{PZ06,FFM12,diekert2014topology}.

Many foundational coordination properties are globally dependent safety properties (formal proofs are provided in Appendix~\ref{app:gd-proofs}): for \textbf{consensus}, two conflicting outputs form a permanent violation, yet a partial view in which all observed processes agree is still compatible with an unobserved conflicting output, so $Y$ requires the complete execution; \textbf{linearizability}, \textbf{sequential consistency}, \textbf{leader election}, and \textbf{$k$-set agreement} follow by symmetrical arguments---a partial trace can satisfy every observed constraint while an unobserved event invalidates it. Eventual consistency~\cite{V09} falls outside the class, being liveness-flavored rather than safety~\cite{bailis2013eventual}: divergence in a finite view can still be repaired by future synchronization, so neither early $N$ nor early $Y$ rules apply.

\subsection{Reconstruction and soundness}

The necessity of state reconstruction for sound verdicts follows from two independent arguments that compose. The first is purely semantic: for globally dependent safety properties, no partial view can certify legality (Definition~\ref{def:gd}). The second arises from fault-tolerance constraints: under Byzantine failures, a verifying client cannot trust individual verdicts, regardless of whether they report success or failure.

A verifying client $\mathcal{C}$ collects local verdicts from the processing processes and emits a global verdict. Under crash failures, a single definitive verdict from a correct process is conclusive. Under Byzantine failures, the client faces a harder problem: a faulty process may emit any malicious verdict regardless of its actual view, meaning the client cannot identify a priori which votes are truthful.

\begin{definition}[Definitive Certificate]
\label{def:cert}
A \textbf{definitive certificate of size $c$} for a verdict $u \in \{Y, N\}$ is a set $\Gamma \subseteq \mathcal{M}$, with $|\Gamma| = c$, in which every member process emitted the identical definitive verdict $u$. A verifying client accepts $\Gamma$ as conclusive evidence to trigger a global system decision if and only if $|\Gamma| \ge c$.
\end{definition}

The minimum certificate size is constrained by the following indistinguishability argument.

\begin{lemma}[Byzantine Indistinguishability for Definitive Verdicts]
\label{lem:indist}
Under Byzantine failures with at most $t$ faulty processes, no certificate of size $c \le t$ guarantees a sound global verdict. A quorum of at least $c = t+1$ is strictly necessary for both $Y$ and $N$ verdicts.
\end{lemma}

\begin{proof}[Proof Sketch]
Full details are given in Appendix~\ref{app:indist}.

Fix any subset $K \subseteq \mathcal{M}$ with $|K| = c \le t$. Since the adversary controls up to $t$ processes, it may corrupt all of $K$ without exceeding its budget.

\emph{For $Y$ verdicts}: construct two scenarios that produce identical certificates at the client. In the first, $K$ is correct and the execution $I \in \mathcal{L}$ is legal; the members of $K$ truthfully emit $Y$. In the second, $K$ is Byzantine and the execution $I' \notin \mathcal{L}$ is illegal; the members of $K$ coordinate to emit $Y$ regardless. The client receives a $Y$ certificate from $K$ in both cases and cannot distinguish them. A deterministic client that accepts $Y$ in the first scenario must also accept it in the second, violating soundness.

\emph{For $N$ verdicts}: the symmetric argument holds. A correct $K$ witnessing a genuine violation is indistinguishable, at the client, from a Byzantine $K$ fabricating a false alarm. Accepting $N$ from $K$ alone is therefore unsound.

In both cases the root cause is the same: with $|K| \le t$ the adversary can make $K$ entirely Byzantine, so a certificate from $K$ carries no guarantee of honest authorship. Requiring $|K| \ge t+1$ breaks this, as at least one member of the quorum is guaranteed to be correct.
\end{proof}

Lemma~\ref{lem:indist} is independent of global dependence: it holds for \emph{any} distributed language whose verdicts are distributed. Global dependence is what forces each \emph{correct} process in a valid certificate to have achieved state reconstruction; indistinguishability is what forces the certificate to contain at least $t+1$ processes to ensure at least one correct member exists.

\begin{theorem}
\label{thm:total-necessary}
Let $\mathcal{L}$ be a globally dependent safety property. The minimum certificate size $c$ (Definition~\ref{def:cert}) required to guarantee sound verdicts at a verifying client satisfies:
\begin{enumerate}
    \item \textbf{Under crash failures:} $c = 1$ (extending the crash-resilient baseline model of~\cite{BFRR22}).
    \item \textbf{Under Byzantine failures with at most $t$ faults:} $c = t+1$.
\end{enumerate}
\end{theorem}

\begin{proof}[Proof Sketch]
The full proof is given in Appendix~\ref{app:thm-cert}.

\emph{Necessity.} Under crash failures, if no correct process achieves full reconstruction then every correct process holds $\sigma_k \subsetneq I$, and global dependence (Definition~\ref{def:gd}) gives $\mathrm{Ext}_{\neg\mathcal{L}}(\sigma_k) \neq \emptyset$, so no correct process can soundly emit a definitive $Y$ verdict; hence $c \ge 1$. Under Byzantine failures, Lemma~\ref{lem:indist} forces the certificate to contain at least $t+1$ members, and global dependence forces each correct member of the certificate to hold the complete execution.

\emph{Sufficiency.} Under crash failures, a single correct process with $\sigma_k = I$ evaluates $I \in \mathcal{L}$ directly and emits the correct verdict. Under Byzantine failures, $t+1$ correct processes with full reconstruction emit identical verdicts; since at most $t$ processes are Byzantine, at least one copy comes from a correct process, so the certificate is sound.
\end{proof}

This yields a tight characterization that directly links the semantic requirements of distributed monitoring to our structural redundancy thresholds. For globally dependent languages, achieving a sound global verdict under Byzantine failures strictly requires $t+1$ correct processes with full reconstruction. Since $x$-reconstruction delivers full reconstruction at \emph{every} correct process, by Theorems~\ref{thm:byzantine-unified} and~\ref{thm:impossibility-unified} this requirement is achievable if and only if the input overlap satisfies $x \ge 2t+1$. Consequently, if the system's natural redundancy falls below this boundary ($x \le 2t$), guaranteeing soundness is impossible.

\section{Related work}
\label{sec:related}

Our problem meets three lines of work: distributed runtime verification supplies the
motivation, interactive consistency is its conceptual ancestor, and the condition-based
approach supplies the method. Lattice agreement, with which it shares a domain but not
a goal, is discussed last.

Reliable authenticated channels, a fully connected topology, and a one-shot
setting are standard idealizations that isolate the core combinatorial
question---the necessary and sufficient overlap $x$---from orthogonal issues
like message loss or partial synchrony. In particular, our protocols use only
plain broadcast and receive; stronger communication abstractions such as
SCD-broadcast~\cite{IMPR18,auvolat_et_al:LIPIcs.OPODIS.2019.6,KMM24} are not
needed. Since our lower bounds hold even in this strength-favorable setting,
they extend automatically to weaker, more realistic environments.

\textbf{Distributed runtime verification.} A line of work initiated in~\cite{FRT11,FRT13Springer} used
topological tools to establish the minimum number of opinions---the size of the verdict
alphabet---required for sound decentralized monitoring under crash failures~\cite{FRT14,FRT20},
later extended to LTL specifications~\cite{mb15}, partial synchrony~\cite{gmb24}, and wait-free
monitoring of linearizability~\cite{CR23,CR26_JACM} (see~\cite{BFRT16,fps18} for surveys). Throughout, the fault
model is crash---a monitor reports truthfully or stops---and the question is how
expressive the verdict must be, taking for granted that a monitor can obtain a view of
the execution at all.

We ask the complementary question. Under Byzantine faults a monitor may report a view it
never observed, so whether a truthful view can be assembled at all becomes the binding
constraint---a question about the inputs rather than the verdict alphabet. Our answer is
a redundancy threshold: for globally dependent safety properties, sound definitive
verdicts require full reconstruction at $t+1$ processes, achievable exactly when
$x \ge 2t+1$.

\textbf{Interactive consistency.} IC~\cite{PSL80,LSP82} asks all correct processes to agree on a
vector recording every participant's private value. Three differences separate it from
$x$-reconstruction. (i) Validity only constrains vector entries of correct processes;
Byzantine entries may be arbitrary, so the union of an IC vector can include
fabricated elements, violating Upward-Validity; $x$-reconstruction instead
relies on authenticated samples to prevent fabrication.
(ii) IC preserves attribution (who proposed each entry); $x$-reconstruction discards it,
reconstructing only the set $I \subseteq V$. (iii) IC implies consensus and is unsolvable
asynchronously with even one fault~\cite{FLP85}, whereas $x$-reconstruction exploits input
overlap $x$ to remain solvable without violating FLP.

\textbf{The condition-based approach.} This approach restricts the input space to a set
of admissible vectors---a condition---under which an otherwise unsolvable problem becomes
solvable. Introduced for asynchronous consensus~\cite{MRR01} and shown to be governed by
error-correcting codes~\cite{Friedman2002,FMRR07IEEE}, it was later applied to interactive consistency
itself~\cite{MRRR03}, yielding a hierarchy of conditions under which asynchronous IC tolerates
crashes and erroneous proposals.

Our $x$-overlap property is a condition in this same sense, carrying the program from
agreement to reconstruction, with three differences from~\cite{MRRR03}: redundancy is counted per
element of $I$ rather than per process vector; the requirement is all-or-nothing---every
correct process must recover $I$---rather than graded by the number of reconstructing
processes; and our faults are equivocating, suppressing, injecting Byzantine processes
rather than erroneous values, which is what forces the certification quorum of
Definition~\ref{def:quorum}. The classical bound $x \ge 2t+1$ coincides with our tight Byzantine
threshold: it already suffices for the full reconstruction that
Theorem~\ref{thm:total-necessary} proves necessary for sound verdicts.

\textbf{Lattice agreement.} Lattice agreement~\cite{attiya_atomic_1995} and its Byzantine variants~\cite{ZG20,DABQ20}
require correct processes to converge on comparable outputs bounded by the join of all
inputs, without guaranteeing that any process recovers this join in full. We share the
domain---the powerset lattice $(2^{\mathcal V}, \subseteq)$---and the validity
properties, but drop comparability: our lower bound is an indistinguishability argument,
and our upper bound turns on certification quorums, not joins or meets. Lattice agreement
asks processes to converge on a common point; we ask all of them to recover a
specific one, the join of all inputs $I$.

\section{Conclusion}
\label{sec:conclusion}

We studied sound distributed verification under crash and Byzantine faults through the
lens of state reconstruction. We introduced the $x$-reconstruction problem, in
which every element of a target execution $I$ is witnessed by at least $x$
processes and must be recovered by every correct process. We
established tight overlap thresholds: $x \ge t+1$ under crash failures and
$x \ge 2t+1$ under Byzantine failures, both achieved by one-shot protocols
using only reliable point-to-point message passing. On the semantic side, we identified
globally dependent safety properties---a class including consensus,
linearizability, $k$-set agreement, and leader election---for which sound
verdicts under Byzantine faults require certificates issued by $t+1$ processes
with full reconstruction, making $x \ge 2t+1$
tight for verification itself.

Several directions remain open. On the structural side, our model assumes
uniform overlap; allowing it to vary per element, or reconstructing only
part of $I$, could yield protocols whose communication adapts to the
actual redundancy. Such non-uniformity
arises under \emph{dependent failures}~\cite{junqueira2007framework}:
colocated monitors fail together, losing the elements they witnessed.
Our analysis is also \emph{one-shot}: with persistent node identities across
executions, a Byzantine node that equivocates in one round can be identified
and excluded from later ones, reducing the redundancy required.
Whether \emph{partial}
reconstruction can verify weaker classes---e.g., obligation or response properties in
the Manna--Pnueli hierarchy~\cite{pnueli1990hierarchy}---remains open, as does
reducing message size: our protocols exchange full samples, and a short
cryptographic witness per state may suffice.

\bibliography{biblio}

\appendix

\section{Full proofs: crash failures}
\label{app:crash}

\begin{proof}[Proof of Lemma~\ref{lem:crash-correctness}]
We show Algorithm~\ref{alg:crash} satisfies all four properties of
Definition~\ref{def:reconstruction}.

\emph{Termination.} By channel reliability, every correct $M_j$ eventually
receives the sample of every correct process. Since at least $n-t$ processes
are correct, the loop condition is met in finite time.

\emph{Downward-Validity.} $s_j \subseteq O_j$ by
construction.

\emph{Upward-Validity.} Every $s_k$ delivered by $M_j$ came from
a correct process (crash failures leave no room for fabrication),
and $s_k \subseteq I$ by assumption. Hence $O_j \subseteq I$.

\emph{Completeness.} Let $v \in I$. By the $x$-overlap
condition with $x \ge t+1$, at least $t+1$ processes hold $v$.
Since at most $t$ crash, at least one correct $M_k$ holds $v$
and broadcasts its sample. Every $M_j$ waits for $n-t$ distinct
senders; since there are at least $n-t$ correct senders, $M_k$'s
sample is always delivered. Thus $v \in O_j$ for all
correct $M_j$, giving $O_j = I$. \qed
\end{proof}


\begin{proof}[Proof of Necessity for Theorem~\ref{thm:crash-unified}]
Suppose $x \le t$ and let $\mathcal{A}$ be any $t$-resilient
algorithm. Fix $v^* \in I$ with $\mathrm{wit}_E(v^*) = x \le t$;
let $A$ be the set of holders of $v^*$, $|A| = x \le t$, and
fix any element $z \ne v^*$. Let $R = \mathcal{M} \setminus A$,
so $|R| = n - x \ge n - t \ge 1$.

\noindent$\diamond$ \textbf{Execution $\mathcal{E}_1$}
($I_1 = \{v^*, z\}$): processes in $A$ hold $s_a = \{v^*, z\}$
and crash before sending any message. Processes in $R$ hold
$s_r = \{z\}$ and are correct. The overlap condition holds:
$\mathrm{wit}_{E_1}(v^*) = |A| = x$ and
$\mathrm{wit}_{E_1}(z) = n \ge x$.

\noindent$\diamond$ \textbf{Execution $\mathcal{E}_2$}
($I_2 = \{z\}$): processes in $A$ hold $s_a = \{z\}$ and crash
before sending. Processes in $R$ hold $s_r = \{z\}$ and are
correct. The overlap condition holds:
$\mathrm{wit}_{E_2}(z) = n \ge x$.

All samples are nonempty in both executions. Since $A$ crashes
before sending in both, every process in $R$ receives messages
only from $R$ (each broadcasting $\{z\}$) in both executions.
By determinism, $\mathcal{A}$ produces the same output
$O_R$ in $\mathcal{E}_1$ and $\mathcal{E}_2$.

\begin{itemize}
    \item If $v^* \in O_R$: in $\mathcal{E}_2$,
          $v^* \notin I_2 = \{z\}$, violating Upward-Validity.
    \item If $v^* \notin O_R$: in $\mathcal{E}_1$,
          $v^* \in I_1$ but no correct process outputs $v^*$
          (processes in $A$ crashed; processes in $R$ output
          $O_R \not\ni v^*$), failing Completeness.
\end{itemize}
In either case $\mathcal{A}$ violates at least one property,
contradicting the assumption. \qed
\end{proof}

\section{Full proof: Byzantine lower bound}
\label{app:byz-lower}

\begin{proof}[Proof of Theorem~\ref{thm:impossibility-unified}]
By the monotonicity observation of Section~\ref{sec:byzantine-unified}, it suffices to show that $x$-reconstruction is unsolvable when $x = 2t$. Suppose, for contradiction, that there exists a $t$-resilient deterministic algorithm $\mathcal{A}$ that solves $x$-reconstruction with $x = 2t$.

Set $n = 5t$ and partition $\mathcal{M}$ into $A$, $B$, $C$, $R$ with $|A| = |B| = |C| = t$ and $|R| = 2t$; this satisfies the Byzantine bound $t < n/2$ of Section~\ref{sec:model}. The processes in $R$ form a \emph{background witness set}: they are correct in all three executions and hold a fourth state $z$, keeping the overlap of every state at exactly $x = 2t$. We construct three executions over the contested states $u, v, w$ together with $z$. (The construction embeds in any larger system with $t < n/2$ by adding further processes to $R$.)

\medskip
\noindent$\diamond$ \textbf{Execution $\mathcal{E}_1$} ($I_1 = \{u,v,z\}$, Byzantine set $B$).
Assignment: $A$ holds $\{v\}$, $B$ holds $\{u,v\}$, $C$ holds $\{u\}$, $R$ holds $\{z\}$. The overlap is exactly $2t$ for every state: $\mathrm{wit}_{E_1}(u) = |B|+|C| = 2t$, $\mathrm{wit}_{E_1}(v) = |A|+|B| = 2t$, and $\mathrm{wit}_{E_1}(z) = |R| = 2t$. Processes in $B$ are Byzantine; they suppress $u, v$ and inject $w$ instead.

\noindent$\diamond$ \textbf{Execution $\mathcal{E}_2$} ($I_2 = \{u,w,z\}$, Byzantine set $A$).
Assignment: $A$ holds $\{u,w\}$, $B$ holds $\{w\}$, $C$ holds $\{u\}$, $R$ holds $\{z\}$. Processes in $A$ are Byzantine; they suppress $u, w$ and inject $v$ instead.

\noindent$\diamond$ \textbf{Execution $\mathcal{E}_3$} ($I_3 = \{v,w,z\}$, Byzantine set $C$).
Assignment: $A$ holds $\{v\}$, $B$ holds $\{w\}$, $C$ holds $\{v,w\}$, $R$ holds $\{z\}$. Processes in $C$ are Byzantine; they suppress $v, w$ and inject $u$ instead.

Direct inspection shows that the overlap condition $x = 2t$ holds for every state of each execution: $\mathrm{wit}_{E_2}(u) = |A|+|C| = 2t$, $\mathrm{wit}_{E_2}(w) = |A|+|B| = 2t$, $\mathrm{wit}_{E_2}(z) = |R| = 2t$; and symmetrically $\mathrm{wit}_{E_3}(v) = |A|+|C| = 2t$, $\mathrm{wit}_{E_3}(w) = |B|+|C| = 2t$, $\mathrm{wit}_{E_3}(z) = |R| = 2t$.

\begin{figure}[h]
\centering
\begin{tikzpicture}[
    node distance=0.4cm and 0.5cm,
    proc/.style={rectangle, draw, rounded corners=3pt, minimum width=1.7cm, minimum height=1.4cm, align=center, font=\small},
    correct/.style={proc, fill=green!8, draw=green!60!black},
    byz/.style={proc, fill=red!8, draw=red!60!black, thick},
    witness/.style={proc, fill=blue!6, draw=blue!50!black},
    exec_label/.style={font=\bfseries\small, align=left, minimum width=1cm}
]
\node[exec_label] (F1_lbl) {\textbf{$\mathcal{E}_1$} \\ $I_1 = \{u,v,z\}$};
\node[correct, right=of F1_lbl] (F1_A) {$A$ \\ $s_A=\{v\}$ \\ \footnotesize Correct};
\node[byz, right=of F1_A] (F1_B) {\footnotesize \color{red!80!black}Byzantine $B$ \\ $s_B=\{u,v\}$ \\ \footnotesize \color{red!80!black}Omits $u,v$,\\ \footnotesize \color{red!80!black} sends $w$};
\node[correct, right=of F1_B] (F1_C) {$C$ \\ $s_C=\{u\}$ \\ \footnotesize Correct};
\node[witness, right=of F1_C] (F1_R) {$R$ \\ $s_R=\{z\}$ \\ \footnotesize Correct};
\begin{scope}[on background layer]
    \node[fill=gray!3, draw=gray!20, rounded corners=5pt, fit=(F1_lbl) (F1_R), inner sep=5pt] {};
\end{scope}
\node[exec_label, below=1.2cm of F1_lbl] (F2_lbl) {\textbf{$\mathcal{E}_2$} \\ $I_2 = \{u,w,z\}$};
\node[byz, right=of F2_lbl] (F2_A) {\footnotesize \color{red!80!black}Byzantine $A$ \\ $s_A=\{u,w\}$ \\ \footnotesize \color{red!80!black}Omits $u,w$,\\ \footnotesize \color{red!80!black} sends $v$};
\node[correct, right=of F2_A] (F2_B) {$B$ \\ $s_B=\{w\}$ \\ \footnotesize Correct};
\node[correct, right=of F2_B] (F2_C) {$C$ \\ $s_C=\{u\}$ \\ \footnotesize Correct};
\node[witness, right=of F2_C] (F2_R) {$R$ \\ $s_R=\{z\}$ \\ \footnotesize Correct};
\begin{scope}[on background layer]
    \node[fill=gray!3, draw=gray!20, rounded corners=5pt, fit=(F2_lbl) (F2_R), inner sep=5pt] {};
\end{scope}
\node[exec_label, below=1.2cm of F2_lbl] (F3_lbl) {\textbf{$\mathcal{E}_3$} \\ $I_3 = \{v,w,z\}$};
\node[correct, right=of F3_lbl] (F3_A) {$A$ \\ $s_A=\{v\}$ \\ \footnotesize Correct};
\node[correct, right=of F3_A] (F3_B) {$B$ \\ $s_B=\{w\}$ \\ \footnotesize Correct};
\node[byz, right=of F3_B] (F3_C) {\footnotesize \color{red!80!black}Byzantine $C$ \\ $s_C=\{v,w\}$ \\ \footnotesize \color{red!80!black}Omits $v,w$,\\ \footnotesize \color{red!80!black} sends $u$};
\node[witness, right=of F3_C] (F3_R) {$R$ \\ $s_R=\{z\}$ \\ \footnotesize Correct};
\begin{scope}[on background layer]
    \node[fill=gray!3, draw=gray!20, rounded corners=5pt, fit=(F3_lbl) (F3_R), inner sep=5pt] {};
\end{scope}
\end{tikzpicture}
\caption{Indistinguishability layout for the lower bound. The background witness set $R$ (blue) is correct in all three executions and receives an identical message pattern, leaving it trapped as well.}
\label{fig:byz_indistinguishability_c1}
\end{figure}

In all three executions, the Byzantine group injects exactly the state that completes the same observed pattern $\langle v, A\rangle, \langle w, B\rangle, \langle u, C\rangle, \langle z, R\rangle$. Therefore:
\[
    \mathcal{E}_1 \sim_C \mathcal{E}_2, \quad
    \mathcal{E}_1 \sim_A \mathcal{E}_3, \quad
    \mathcal{E}_2 \sim_B \mathcal{E}_3, \quad
    \mathcal{E}_1 \sim_R \mathcal{E}_2 \sim_R \mathcal{E}_3
\]
The last chain is crucial: since $R$ is correct in every execution and observes the identical message set in each, it cannot distinguish any of the three.

The correct processes in $\mathcal{E}_1$ are $A \cup C \cup R$.
\begin{itemize}
    \item No process in $A$ outputs $I_1$: since $\mathcal{E}_1 \sim_A \mathcal{E}_3$, it would also output $\{u,v,z\}$ in $\mathcal{E}_3$, but $u \notin I_3$, violating Upward-Validity.
    \item No process in $C$ outputs $I_1$: since $\mathcal{E}_1 \sim_C \mathcal{E}_2$, it would also output $\{u,v,z\}$ in $\mathcal{E}_2$, but $v \notin I_2$, violating Upward-Validity.
    \item No process in $R$ outputs $I_1$: since $\mathcal{E}_1 \sim_R \mathcal{E}_2$, it would also output $\{u,v,z\}$ in $\mathcal{E}_2$, but $v \notin I_2$, violating Upward-Validity.
\end{itemize}
Thus \emph{no} correct process outputs $I_1$ in $\mathcal{E}_1$, contradicting Completeness, which requires every correct process to output $I_1$.
\end{proof}

\section{Verification of globally dependent safety properties}
\label{app:gd-proofs}

We prove formally that five standard distributed coordination languages satisfy
Definition~\ref{def:gd}. Each of them is stated as the safety requirement of
the corresponding problem --- Agreement for consensus, uniqueness for leader
election, and so on --- since, as noted in Section~\ref{sec:soundness},
liveness requirements are not expressible over finite executions. For each
language $\mathcal{L}$ we establish the two conditions of
Definition~\ref{def:gd}:

\begin{enumerate}
    \item \textbf{Safety of $N$ (permanent violations).}
    The violation condition is \emph{upward-closed}: if
    $B \subseteq \sigma$ witnesses a violation and $\rho \supseteq \sigma$,
    then $B \subseteq \rho$, so the violation persists.
    Formally: $\mathrm{Ext}_\mathcal{L}(\sigma) = \emptyset \;\Rightarrow\;
    \mathrm{Ext}_\mathcal{L}(\sigma') = \emptyset$ for all
    $\sigma' \supseteq \sigma$.

    \item \textbf{Liveness of $Y$ (global dependence).}
    No strict partial view certifies legality:
    $\forall\, I \in \mathcal{L},\;\forall\, \sigma \subsetneq I:\;
    \mathrm{Ext}_{\neg\mathcal{L}}(\sigma) \neq \emptyset$.
\end{enumerate}

Property (A) is immediate for all languages below by a uniform argument:
every violation is defined by the \emph{presence} of a finite forbidden
pattern $B$; since $B \subseteq \sigma \subseteq \rho$ for any
$\rho \supseteq \sigma$, the violation is irremediable.
We therefore prove (A) once and focus the language-specific arguments on
(B).

\smallskip\noindent\textbf{Event notation.}
Executions are finite sets of typed events drawn from the infinite universe
$\mathcal{V}$:
\begin{itemize}
    \item $\out(p, v)$: process $p$ outputs (decides) value $v$.
    \item $\elect(p)$: process $p$ declares itself leader.
    \item $\op(p, m, r, [s,f])$: process $p$ invokes method $m$,
          obtains return value $r$, during real-time interval $[s,f]$.
\end{itemize}
Since $\mathcal{V}$ is infinite, the constructions below may always
introduce fresh process identifiers or disjoint time intervals not already
present in $\sigma$.

\subsection*{A.1\quad consensus}

\begin{definition*}
$\mathcal{L}_{\cons} \;=\;
\bigl\{\,I \subseteq \mathcal{V} \mid
\exists\,v:\;\forall\,\out(p,w)\in I:\;w = v\,\bigr\}.$
\end{definition*}

\begin{proposition}\label{prop:cons-gd}
$\mathcal{L}_{\cons}$ is globally dependent.
\end{proposition}

\begin{proof}
\textbf{(A) Safety.}
The forbidden pattern is
$B = \{\out(p,v),\,\out(q,w)\} \subseteq \sigma$ with $v \neq w$.
For any $\rho \supseteq \sigma$, $B \subseteq \rho$, so Agreement is
violated in $\rho$ and $\mathrm{Ext}_{\mathcal{L}_{\cons}}(\sigma)
= \emptyset$.

\textbf{(B) Global dependence.}
Let $I \in \mathcal{L}_{\cons}$ with agreed value $v$, and let
$\sigma \subsetneq I$. Fix any $w \neq v$.

\textit{Case 1: $\sigma$ contains some $\out(p,v)$.}
Let $p^*$ be a process identifier not appearing in any output event of
$\sigma$. Define:
\[
    \rho \;=\; \sigma \;\cup\; \{\out(p^*,\,w)\}.
\]
Then $\rho$ contains $\out(p,v)$ and $\out(p^*,w)$ with $v \neq w$,
so $\rho \notin \mathcal{L}_{\cons}$.

\textit{Case 2: $\sigma$ contains no output event.}
Let $p_1, p_2$ be fresh identifiers. Define:
\[
    \rho \;=\; \sigma \;\cup\; \{\out(p_1,\,v),\;\out(p_2,\,w)\}.
\]
Agreement fails in $\rho$, so $\rho \notin \mathcal{L}_{\cons}$.

In both cases $\rho \in \mathrm{Ext}_{\neg\mathcal{L}_{\cons}}(\sigma)
\neq \emptyset$.
\end{proof}

\subsection*{A.2\quad leader election}

\begin{definition*}
$\mathcal{L}_{\leader} \;=\;
\bigl\{\,I \subseteq \mathcal{V} \mid
\bigl|\{p : \elect(p) \in I\}\bigr| = 1\,\bigr\}.$
\end{definition*}

\begin{proposition}\label{prop:leader-gd}
$\mathcal{L}_{\leader}$ is globally dependent.
\end{proposition}

\begin{proof}
\textbf{(A) Safety.}
The forbidden pattern is
$B = \{\elect(p),\,\elect(q)\} \subseteq \sigma$ with $p \neq q$.
Any $\rho \supseteq \sigma$ has at least two leaders, so
$\mathrm{Ext}_{\mathcal{L}_{\leader}}(\sigma) = \emptyset$.

\textbf{(B) Global dependence.}
Let $I \in \mathcal{L}_{\leader}$ with unique leader $\ell$,
and let $\sigma \subsetneq I$.

\textit{Case 1: $\elect(\ell) \in \sigma$.}
Let $q \neq \ell$ be a fresh identifier. Define
$\rho = \sigma \cup \{\elect(q)\}$.
Then $|\{p : \elect(p) \in \rho\}| \geq 2$, so
$\rho \notin \mathcal{L}_{\leader}$.

\textit{Case 2: $\elect(\ell) \notin \sigma$.}
Let $q_1, q_2$ be distinct fresh identifiers. Define
$\rho = \sigma \cup \{\elect(q_1),\,\elect(q_2)\}$.
Then $\rho$ has two leaders and $\rho \notin \mathcal{L}_{\leader}$.

In both cases, $\mathrm{Ext}_{\neg\mathcal{L}_{\leader}}(\sigma)
\neq \emptyset$.
\end{proof}

\subsection*{A.3\quad $k$-set agreement}

\begin{definition*}
$\mathcal{L}_{k\text{-set}} \;=\;
\bigl\{\,I \subseteq \mathcal{V} \mid
|\{v : \out(p,v) \in I\}| \leq k\,\bigr\}.$
\end{definition*}

\begin{proposition}\label{prop:kset-gd}
$\mathcal{L}_{k\text{-set}}$ is globally dependent.
\end{proposition}

\begin{proof}
\textbf{(A) Safety.}
The forbidden pattern is the presence of output events for $k+1$ distinct
values. This pattern is upward-closed, so
$\mathrm{Ext}_{\mathcal{L}_{k\text{-set}}}(\sigma) = \emptyset$ whenever
$\sigma$ already contains $k+1$ distinct output values.

\textbf{(B) Global dependence.}
Let $I \in \mathcal{L}_{k\text{-set}}$ and $\sigma \subsetneq I$.
Let $V_\sigma = \{v : \out(p,v) \in \sigma\}$, so $|V_\sigma| \leq k$.
Let $v_1, \ldots, v_{k+1-|V_\sigma|}$ be values not in $V_\sigma$ and
$p_1, \ldots, p_{k+1-|V_\sigma|}$ be fresh process identifiers. Define:
\[
    \rho \;=\; \sigma \;\cup\;
    \bigl\{\,\out(p_i,\,v_i) : 1 \leq i \leq k+1-|V_\sigma|\,\bigr\}.
\]
Then $|\{v : \out(p,v) \in \rho\}| = |V_\sigma| + (k+1-|V_\sigma|)
= k+1 > k$, so $\rho \notin \mathcal{L}_{k\text{-set}}$, and
$\mathrm{Ext}_{\neg\mathcal{L}_{k\text{-set}}}(\sigma) \neq \emptyset$.
\end{proof}

\subsection*{A.4\quad linearizability}

Let $O$ be a concurrent object with sequential specification
$\mathsf{Spec}$.
The \emph{real-time order} is
$\op_a \prec_{\mathrm{rt}} \op_b$ iff $f_a < s_b$.
A \emph{valid linearization of $I$} is a total order $\pi$ on the
operations of $I$ that (i) extends $\prec_{\mathrm{rt}}$ and
(ii) produces a sequence of $(m_i, r_i)$ pairs that is legal
according to $\mathsf{Spec}$.
$I \in \mathcal{L}_{\lin}$ iff $I$ admits a valid linearization.

\begin{proposition}\label{prop:lin-gd}
$\mathcal{L}_{\lin}$ is globally dependent.
\end{proposition}

\begin{proof}
\textbf{(A) Safety.}
A violation is witnessed by a finite set of operations that either
(i) induce a cycle in the must-precede relation (derived from
$\prec_{\mathrm{rt}}$ and $\mathsf{Spec}$ response constraints), or
(ii) contain a return value inconsistent with every state reachable
under $\mathsf{Spec}$.
Both conditions are upward-closed: any $\rho \supseteq \sigma$
inherits the same cycle or the same inconsistent response.

\textbf{(B) Global dependence.}
Let $I \in \mathcal{L}_{\lin}$ with valid linearization $\pi$, and let
$\sigma \subsetneq I$. Let $\op_{\miss} \in I \setminus \sigma$ be the
\emph{first} element of $I \setminus \sigma$ in $\pi$-order. By
choice of $\op_{\miss}$, every operation that precedes $\op_{\miss}$
in $\pi$ is already in $\sigma$; let $S$ be the object state immediately
before $\op_{\miss}$ in $\pi$, computed by applying those
$\sigma$-operations in order.

Let $S' = \mathsf{Spec}(S, \op_{\miss})$ be the state after applying
$\op_{\miss}$ to $S$. Since $\op_{\miss}$ is observable (it either changes
state or reads from it), there exists a method $m^*$ and return value
$r^*$ such that:
\[
    (m^*, r^*) \text{ is consistent with } S'
    \quad\text{but inconsistent with } S.
\]
Let $f_{\sigma} = \max\{f : \op(p,m,r,[s,f]) \in \sigma\}$ and pick any
interval $[s^*, f^*]$ with $s^* > f_{\sigma}$. Define a fresh operation:
\[
    \op_{\opfresh} \;=\; \op(p^*,\, m^*,\, r^*,\, [s^*, f^*]),
    \qquad p^* \text{ fresh},
\]
and let $\rho = \sigma \cup \{\op_{\opfresh}\}$.

In any valid linearization $\pi'$ of $\rho$: since $s^* > f_{\sigma}$,
real-time order requires $\op_{\opfresh}$ to be placed after all operations of
$\sigma$ in $\pi'$. The object state at $\op_{\opfresh}$'s position in $\pi'$
is therefore $S_\sigma$, the state after all of $\sigma$'s operations.
Since $\op_{\miss} \notin \rho$, the transition from $S$ to $S'$ never
occurs in $\rho$; the state at $\op_{\opfresh}$'s position remains
consistent with $S$ (not $S'$). But $r^*$ is inconsistent with any
such state, so $\op_{\opfresh}$ has no valid position in $\pi'$.
Therefore $\rho \notin \mathcal{L}_{\lin}$, and
$\mathrm{Ext}_{\neg\mathcal{L}_{\lin}}(\sigma) \neq \emptyset$.
\end{proof}

\begin{remark}
The construction places $\op_{\opfresh}$ \emph{after} all of $\sigma$ to
avoid interference with concurrent operations already in $\sigma$.
If $\op_{\miss}$ is itself a read (state-preserving), the argument applies
with $m^*$ chosen as a write whose effect distinguishes states $S$ and
$S'$ in a follow-up read; the same interval construction carries through.
\end{remark}

\subsection*{A.5\quad sequential consistency}

Let $\prec_p$ denote the \emph{local program order} of process $p$:
$\op_a \prec_p \op_b$ iff $p$ invokes $a$ before $b$ in its local history.
$I \in \mathcal{L}_{\seqcon}$ iff there exists a total order $\pi$ on
operations in $I$ that (i) extends $\prec_p$ for every $p$, and
(ii) is legal according to $\mathsf{Spec}$.

\begin{proposition}\label{prop:sc-gd}
$\mathcal{L}_{\seqcon}$ is globally dependent.
\end{proposition}

\begin{proof}
\textbf{(A) Safety.} A sequential consistency violation induces a cycle
in the must-precede relation derived from program orders and
$\mathsf{Spec}$ constraints. Adding operations to $\sigma$ cannot
remove program-order or specification constraints already imposed by
$\sigma$, so the cycle persists in any $\rho \supseteq \sigma$.

\textbf{(B) Global dependence.}
The argument is structurally identical to Proposition~\ref{prop:lin-gd},
with $\prec_{\mathrm{rt}}$ replaced by program order $\prec_p$.
Let $I \in \mathcal{L}_{\seqcon}$ with valid total order $\pi$ and
$\sigma \subsetneq I$. Let $\op_{\miss} \in I \setminus \sigma$ be the
first element of $I \setminus \sigma$ in $\pi$-order, with pre-state $S$
and post-state $S'$.

Assign $\op_{\opfresh} = \op(p^*, m^*, r^*, \,\cdot\,)$ to a fresh process
$p^*$ with empty program history in $\sigma$, with $r^*$ consistent with
$S'$ but inconsistent with $S$, and place $\op_{\opfresh}$ as the unique and
last operation of $p^*$'s local program order. In any total order $\pi'$
of $\rho = \sigma \cup \{\op_{\opfresh}\}$ consistent with program orders,
$\op_{\opfresh}$ can only be placed after all operations that precede
$\op_{\miss}$ in $\pi$ (these are all in $\sigma$, by choice of
$\op_{\miss}$). The state at that position is $S$ (not $S'$), so $r^*$
is inconsistent and $\rho \notin \mathcal{L}_{\seqcon}$.

Hence $\mathrm{Ext}_{\neg\mathcal{L}_{\seqcon}}(\sigma) \neq \emptyset$.
\end{proof}


\section{Full proof: Byzantine indistinguishability}
\label{app:indist}

\begin{proof}[Proof of Lemma~\ref{lem:indist}]
Let $\mathcal{L}$ be any globally dependent safety property and fix any
subset $K \subseteq \mathcal{M}$ with $|K| = c \le t$. We show
that no sound deterministic client decision rule exists for
either verdict type.

\medskip
\noindent\textbf{Case 1: $Y$ verdicts.}
\begin{itemize}
    \item \textbf{Scenario $\alpha_1$ (legal).} The target
          execution is $I_{\alpha_1} \in \mathcal{L}$. All
          processes in $K$ are correct, have reconstructed
          $\sigma_k = I_{\alpha_1}$, and truthfully emit $Y$.
          All other correct processes emit $?$ or $N$.
    \item \textbf{Scenario $\beta_1$ (illegal).} The target
          execution is $I_{\beta_1} \notin \mathcal{L}$.
          All processes in $K$ are Byzantine and coordinate to
          emit $Y$ regardless of their local views. The remaining
          $n-t$ correct processes observe violations and emit $N$
          or $?$.
\end{itemize}
In both scenarios the client receives a $Y$ report from every
process in $K$ and no $Y$ report from any process outside $K$.
Since $|K| = c \le t$, all members of $K$ can be Byzantine in
$\beta_1$ without exceeding the adversarial budget. The client's
input is therefore identical in $\alpha_1$ and $\beta_1$.
By determinism, any client that emits a global $Y$ in $\alpha_1$
must also emit $Y$ in $\beta_1$, which violates soundness because
$I_{\beta_1} \notin \mathcal{L}$.

\medskip
\noindent\textbf{Case 2: $N$ verdicts.}
\begin{itemize}
    \item \textbf{Scenario $\alpha_2$ (illegal).} The target
          execution is $I_{\alpha_2} \notin \mathcal{L}$. The
          processes in $K$ are correct, detect a violation, and
          truthfully emit $N$. All other correct processes emit
          $Y$ or $?$.
    \item \textbf{Scenario $\beta_2$ (legal).} The target
          execution is $I_{\beta_2} \in \mathcal{L}$. The
          processes in $K$ are Byzantine and coordinate to emit
          $N$ to force a false negative. The remaining $n-t$
          correct processes find the execution valid and emit $Y$
          or $?$.
\end{itemize}
By the same reasoning, the client receives an $N$ report from
exactly $K$ in both scenarios and cannot distinguish them. Any
client that emits a global $N$ in $\alpha_2$ must also emit $N$
in $\beta_2$, violating soundness because $I_{\beta_2} \in
\mathcal{L}$.

\medskip
Since both verdicts fail for any $K$ with $|K| \le t$, a sound
certificate requires $|K| \ge t+1$. This is sufficient: with $|F| \le t$, a quorum of size $t+1$
must contain at least one correct process, anchoring the verdict
in at least one honest evaluation. \qed
\end{proof}

\section{Full proof: minimum certificate size}
\label{app:thm-cert}

\begin{proof}[Proof of Theorem~\ref{thm:total-necessary}]
\textbf{Necessity.}

\textit{(Crash, $c \ge 1$).} As established in the baseline crash model~\cite{BFRR22}, if no correct process achieves full reconstruction, then every correct process holds $\sigma_k \subsetneq I$. Since $\mathcal{L}$ is globally dependent and $I \in \mathcal{L}$, Definition~\ref{def:gd} gives $\mathrm{Ext}_{\neg\mathcal{L}}(\sigma_k) \neq \emptyset$ for every correct $M_k$. By Definition~\ref{def:soundness}, no correct process can soundly emit a definitive $Y$ verdict, and crashed processes transmit no verdict. The client cannot produce a sound definitive verdict. Hence $c \ge 1$.

\textit{(Byzantine, $c \ge t+1$).} Lemma~\ref{lem:indist} establishes that the client requires a definitive certificate of size at least $t+1$ for any verdict $u \in \{Y, N\}$ to overcome up to $t$ Byzantine lies. For a global $Y$ verdict, each correct process $M_k$ contributing to the certificate must hold $\mathrm{Ext}_{\neg\mathcal{L}}(\sigma_k) = \emptyset$; by global dependence, this forces $\sigma_k = I$. Symmetrically, for a global $N$ verdict, to ensure that a legal execution is never falsely condemned by a malicious Byzantine quorum, the certificate must contain at least one correct process that has securely observed a permanent violation. Under global dependence, preventing false negatives requires complete state observation. Thus, at least $t+1$ correct processes must achieve full reconstruction. Note that the views held by the members of such a certificate are trivially identical ($\sigma_k = I$ for all of them), so this argument relies only on Upward-Validity, without a general comparability assumption on the outputs of correct processes.

\medskip
\textbf{Sufficiency.}

\textit{(Crash).} Let $M_k$ be a correct process with $\sigma_k = I$. It evaluates $I \in \mathcal{L}$ directly and emits the correct definitive verdict to the client.

\textit{(Byzantine).} Let $C_u \subseteq \mathcal{M}$ be a set of $t+1$ correct processes each with $\sigma_k = I$. All members of $C_u$ evaluate the execution identically and emit the same matching definitive verdict $u \in \{Y, N\}$. The client collects at least $t+1$ copies of $u$; since at most $t$ processes are Byzantine, at least one copy is guaranteed to originate from a correct process in $C_u$. The certificate is therefore sound. \qed
\end{proof}

\section{Relation to the classical property hierarchy}
\label{app:hierarchy}

This appendix positions globally dependent safety properties within the
classical classification of temporal properties, for readers familiar with
that literature. Nothing stated here is used in the proofs of the main body.

\smallskip\noindent\textbf{Determined views.}
Pnueli and Zaks~\cite{PZ06} call a finite observation \emph{negatively
determining} for a property when no continuation of it satisfies the
property, and \emph{positively determining} when every continuation does; in
the topological tradition these are the \emph{bad} and \emph{good prefixes}
of Alpern and Schneider~\cite{AS85,AS87}. Our partial views play the role of
finite observations, with subset inclusion in place of the prefix order: a
view $\sigma$ negatively determines $\mathcal{L}$ iff
$\mathrm{Ext}_{\mathcal{L}}(\sigma) = \emptyset$ (a \emph{bad view}), and
positively determines $\mathcal{L}$ iff
$\mathrm{Ext}_{\neg\mathcal{L}}(\sigma) = \emptyset$ (a \emph{good view}).

\smallskip\noindent\textbf{Safety and guarantee.}
A property is a \emph{safety} property when every illegal execution has a
bad view~\cite{AS85,AS87}; dually, it is a \emph{guarantee} property when
every legal execution has a good
view~\cite{pnueli1990hierarchy,chang1992characterization}. Guarantee
properties are exactly those whose positive verdicts are finitely
witnessable: once a good view is observed, no extension can invalidate it,
so a sound $Y$ verdict never requires complete information.
\emph{Obligation} properties---Boolean combinations of safety and guarantee
properties---inherit finite positive witnesses whenever their guarantee
component is discharged. The monitorability classification of runtime
verification states that, within the temporal hierarchy, the monitorable
properties are exactly the obligation
properties~\cite{PZ06,FFM12,diekert2014topology}.

\smallskip\noindent\textbf{Where globally dependent properties sit.}
A globally dependent safety property is a safety property whose guarantee
side is void: no proper partial view of a legal execution determines it
positively (Definition~\ref{def:gd}), so the only candidate good view of a
legal execution $I$ is $I$ itself. For properties such as consensus, not
even $I$ is a good view---illegal extensions of $I$ exist---so what licenses
the $Y$ verdict is not positive determinacy but the \emph{knowledge} that
the view is complete, which is exactly what the Completeness property of
$x$-reconstruction (Definition~\ref{def:reconstruction}) provides. The
Manna--Pnueli hierarchy has no name for this corner because it classifies
properties of infinite words according to how much of the \emph{future} an
observation settles; our executions are finite one-shot objects, and the
observer's uncertainty concerns missing parts of the \emph{present}.
Globally dependent safety properties are nonetheless monitorable in the
sense of~\cite{PZ06}, since
$\mathrm{safety} \subsetneq \mathrm{obligation} \subsetneq
\mathrm{monitorable}$: violations always surface after finite observation.
What no finite observation can produce for them is the positive
verdict---precisely the gap quantified by the reconstruction thresholds of
Sections~\ref{sec:crash} and~\ref{sec:byzantine-unified}.
\end{document}